\documentclass[journal]{IEEEtran}

\usepackage{booktabs}
\usepackage{threeparttable}
\usepackage{amsmath,graphicx,cite,amssymb}
\usepackage{amsfonts,multirow,bm,array,setspace,stfloats}

\usepackage{graphicx,float,cite,amssymb}
\usepackage{graphics}
\DeclareGraphicsExtensions{.pdf,.jpeg,.png,.jpg}   
\usepackage{multirow,bm,bbm,array,setspace}
\usepackage{textcomp}
\usepackage{psfrag}
\usepackage{color}
\usepackage{pstricks,enumerate}
\usepackage{bbm}
\usepackage{subfigure}
\usepackage{xpatch}
\usepackage{algorithm,algpseudocode}
\makeatletter
\newcommand{\distas}[1]{\mathbin{\overset{#1}{\kern\z@\sim}}}%
\newsavebox{\mybox}\newsavebox{\mysim}
\newcommand{\distras}[1]{%
    \savebox{\mybox}{\hbox{\kern1pt$\scriptstyle#1$\kern1pt}}%
	\savebox{\mysim}{\hbox{$\sim$}}%
	\mathbin{\overset{#1}{\kern\z@\resizebox{\wd\mybox}{\ht\mysim}{$\sim$}}}%
}

\ExplSyntaxOn
\cs_new:Npn \bibColoredItems #1#2
  {
    \clist_map_inline:nn {#2} { \cs_new:cpn {bib@colored@##1} {#1} } 
  }
\ExplSyntaxOff

\newcommand\bib@setcolor[1]{%
  \ifcsname bib@colored@#1\endcsname
    \expandafter\color\expandafter{\csname bib@colored@#1\endcsname}
  \else
    \normalcolor
  \fi
}

\xpatchcmd\@bibitem
  {\item}
  {\bib@setcolor{#1}\item}
  {}{\PatchFailed}

\xpatchcmd\@lbibitem
  {\item}
  {\bib@setcolor{#2}\item}
  {}{\PatchFailed}

\makeatother
\ifCLASSINFOpdf
\else
\fi

\newtheorem{proposition}{Proposition}

\newtheorem{theorem}{Theorem}

\newtheorem{remark}{Remark}

\newcommand{\bH}{\bm{H}}

\begin{document}
	%
\title{Blind Beamforming for Intelligent Reflecting Surfaces: A Gradient Bandit Approach}
\author{
\IEEEauthorblockN{
    Wenhai Lai, \IEEEmembership{Member,~IEEE},  Mingxiao Li, \IEEEmembership{Graduate Student Member,~IEEE}, \\and Kaiming Shen, \IEEEmembership{Senior Member,~IEEE}
}
\thanks{
Manuscript submitted on \today. This work has been presented in part at the IEEE International Conference on Acoustics, Speech, and Signal Processing (ICASSP) 2024 \cite{lai2024blind}.

Wenhai Lai is with the Queensland Technology International School of Engineering, Dalian University of Technology, China (e-mail: wenhailai@dlut.edu.cn).

Mingxiao Li and Kaiming Shen are with the School of Science and Engineering, The Chinese University of Hong Kong (Shenzhen), China (e-mails: mingxiaoli@link.cuhk.edu.cn; shenkaiming@cuhk.edu.cn).
}
}

%


\maketitle

\begin{abstract}
The beamforming problem of intelligent reflecting surface (IRS) has been extensively considered from an optimization perspective assuming that channel state information (CSI) is available. However, the reality is that the existing prototypes seldom follow this model-based approach because channel estimation is technically difficult and costly for the network protocols and hardware to date. A recent trend is to perform beamforming blindly without channel knowledge. This work looks at blind beamforming from a reinforcement learning point of view. We first show that the existing blind beamforming method boils down to a special case of the greedy algorithm in the reinforcement learning context. We analyze the resulting cumulative regret, and further propose an upper approximation to facilitate the optimization of the exploration probability. Moreover, we show that a gradient sampling scheme can improve the efficiency of reinforcement learning as compared to the uniform sampling scheme adopted in the existing blind beamforming method. We further verify the convergence of the proposed gradient sampling scheme by showing that it is a stochastic approximation to gradient ascent. Finally, we physically implement the proposed method in a real-world prototype system. Our field test results show that, as compared to the existing blind beamforming method, 
the proposed gradient sampling boosts the average signal-to-noise ratio (SNR) by more than 5 dB with 1000 samples.
\end{abstract}


\section{Introduction}
\label{sec:intro}

\IEEEPARstart{I}{ntelligent} reflecting surface (IRS) \cite{wu2021intelligent,gong2020toward}, a.k.a. reconfigurable intelligent surface (RIS) \cite{huang2019reconfigurable,liu2021reconfigurable}, is an emerging wireless technology that harnesses the signal reflection paths to enhance wireless transmission. An IRS consists of an array of reflective elements (REs)---which induce phase shifts into their respective incident signals. A major use of the IRS is to coordinate phase shifts across its REs, i.e., \emph{passive beamforming}, in order to maximize the received signal power at the target position. The existing approaches to passive beamforming typically entail channel state information (CSI), whereas some recent works \cite{Arun2020RFocus,ren2022configuring,Abdallah2024MADRL,Bin2026adaptive,Wang2025BlindInterference} suggest optimizing phase shifts blindly without channel acquisition. We further develop this new idea by means of reinforcement learning: 
\begin{itemize}
    \item We show that the blind beamforming method in \cite{Arun2020RFocus, ren2022configuring} is equivalent to an $\epsilon$-greedy algorithm for the multi-armed bandit problem \cite{cesabianchi1998finite,auer2002finite,vermorel2005multi}. We analyze the resulting cumulative regret, and further propose an upper approximation to make the optimization of the exploration probability easier.
    \item We improve the efficiency of blind beamforming by replacing uniform samples with gradient samples. Specifically, we show that gradient samples enable faster convergence, allowing better performance to be achieved with fewer samples. Moreover, we demonstrate the advantage of the proposed gradient sampling scheme in field tests.

\end{itemize}

In the literature, passive beamforming has been considered extensively from a model-based optimization perspective: it first estimates the channels for all the propagation paths to formulate the problem model explicitly, and then optimizes the phase shifts based on the problem model. When CSI is available, the phase shift optimization for IRS has been well studied. Many existing works assume that phase shifts can be arbitrarily chosen within the interval $[0,2\pi)$, so the phase shifting problem is a continuous optimization problem. As a result, the passive beamforming problem of IRS can be addressed by the standard methods, e.g., semidefinite relaxation (SDR)
\cite{luo2010semidefinite, zhou2020robust, zeng2020sum, xie2020max, huang2020decentralized} and fractional programming (FP) \cite{shen2018fractional_p1, feng2020physical,zhu2020power,shafique2020optimization,zhang2022active,zhang2021joint}.
However, in practice, the choice of phase shift is typically limited to a discrete set due to the hardware constraint, e.g., each RE takes on its phase shift from the binary set $\{0,\pi\}$ in the 1-bit phase shifting case. To account for the discrete constraint, a common heuristic \cite{Feng2021RFlens,wu2020beamforming,li2020joint,you2020channel} is to relax the problem and then round the continuous solution to the discrete set; the optimal solutions are known only for the single-antenna case \cite{zhang2022optimal,ren2023optimal}.

But estimating channels for an IRS-assisted network is by no means trivial; a variety of sophisticated strategies have been proposed. In \cite{mishra2019channel, yang2020intelligent,lin2020reconfigurable}, the channel estimation is conducted assuming that each RE can be turned off, i.e., every {\footnotesize{OFF}} RE absorbs all signals falling on it.
Another line of work \cite{zheng2019intelligent,jensen2020optimal,de2021channel} requires introducing extra pilots according to the discrete Fourier transform (DFT) matrix. Assuming a certain type of sparse model for channels, \cite{chen2023channel} suggests a two-step method that first estimates the common column of AoD steering vectors and then solves the sparse matrix recovery problem. In particular, we mention that all the above methods entail full cooperation between the IRS and the base station.
However, it is difficult to carry out these channel estimation methods in real-world networks because of the following two practical issues:
\begin{enumerate}[i.]
    \item Each reflected channel alone can be easily overwhelmed by the much stronger direct channel, and hence is difficult to measure.
    \item The existing channel estimation methods for IRS are incompatible with the current networking protocols and hardware.
\end{enumerate}
As such, even though the model-driven approach has been extensively studied in the literature to date, the existing works \cite{Arun2020RFocus, ren2022configuring, pei2021ris, staat2022irshield} that pursue prototype implementations of IRS rarely consider channel estimation.

Due to the difficulties and high cost of channel estimation, there has been an increasing interest in configuring IRS without CSI. For instance, \cite{psomas2021low} proposes a distributed ascent algorithm for the continuous IRS beamforming based on the signal-to-noise ratio (SNR) feedback, and \cite{cai2023toward} proposes an IRS-aided joint index keying $M$-ary differential chaos shift keying system where the receiver can retrieve information bits by performing non-coherent correlation demodulation without CSI. The commonly used beam training method \cite{you2020fast,wang2022fast,wang2021jointbeam} constitutes another line of endeavors to perform blind beamforming, which simply tries out a sequence of random beamformers and then picks the best one. Although the core idea of beam training is fairly simple, the performance of this method is limited by the size of random samples. Another existing idea is to use the neural network to learn the direct mapping from the received pilot signals to the phase shift decision \cite{2022Xuwy,2020HuangcwJSAC,2021JiangtJSAC}. However, the parameter tuning of the neural network can be quite tricky.

\begin{figure*}[t]
    \centering
    \includegraphics[width=0.8\linewidth]{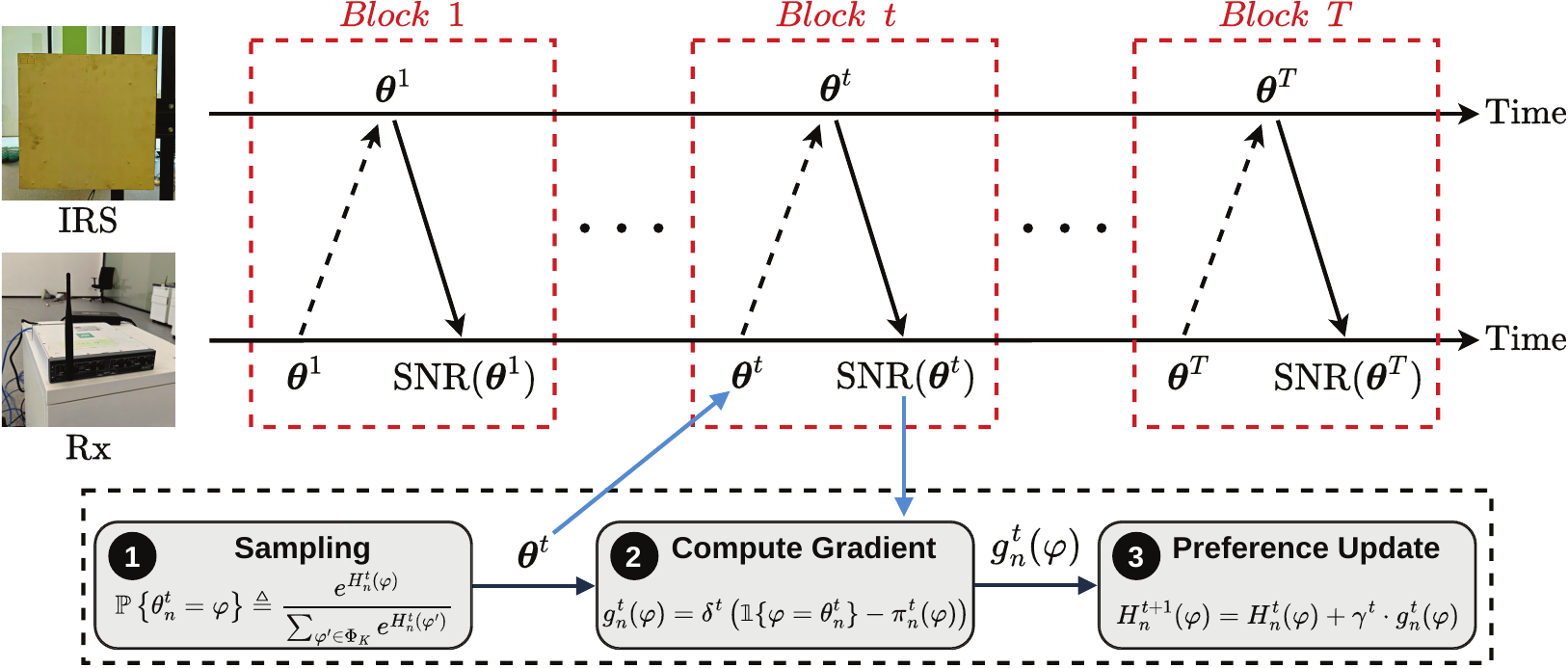}
    \caption{Procedure of the proposed gradient sampling. We divide the channel coherent time into $T$ blocks. Moreover, we introduce a \emph{numerical preference} variable $H_n^t(\varphi)$ for each $(n,\varphi)$ in each block $t$. Specifically, $H_n^t(\varphi)$ characterizes the ``goodness'' of setting $\theta^t_n=\varphi$, where $\theta_n^t$ is the phase shift induced by RE $n$ in block $t$ and $\varphi$ is a possible phase shift choice from a discrete set. In each block $t$, we first generate a phase-shift configuration $\boldsymbol{\theta}^t$ by sampling each RE phase shift $\theta_n$ according to the softmax probability induced by the current preference variables. The receiver measures the resulting SNR as the reward. Based on the SNR feedback, we compute the stochastic gradient $g_n^t(\varphi)$. Finally, the preference variables are updated by stochastic gradient ascent.}
    \label{fig:gradient_sampling}
\end{figure*}

\begin{table}[t]
\caption{List of Main Variables}
\label{tab:main_variables}
\centering
\renewcommand\arraystretch{1.5}
\footnotesize
\begin{tabular}{|c||l|}
\hline
\textbf{Symbol} & \textbf{Definition} \\
\hline\hline

$N$ & number of REs on the IRS \\
\hline

$K$ & number of phase-shift choices for each RE \\
\hline

$T$ & total number of transmission blocks \\
\hline

$\Phi_K$ & $\{0,\omega,\ldots,(K-1)\omega\}\;\;\text{with}\;\;\omega = \frac{2\pi}{K}$ \\
\hline

$\Phi_K^N$ &the $N$-fold Cartesian product of $\Phi_K$\\
\hline

$\theta_n^t$ & phase shift of RE $n$ in block $t$ \\
\hline

$\boldsymbol{\theta}^t$ & IRS phase-shift vector in block $t$
 \\
\hline

$\epsilon$ & exploration probability of the $\epsilon$-greedy algorithm
 \\
\hline

$\mathrm{SNR}(\boldsymbol{\theta}^t)$ & received SNR in block $t$ \\
\hline

$Q_n^t(\varphi)$ & CSM value of phase shift $\varphi$ for RE in block $t$ $n$\\
\hline

$H_n^t(\varphi)$ & preference variable of phase shift $\varphi$ for RE $n$ in block $t$ \\
\hline

$\pi_n^t(\varphi)$ & probability of selecting
$\varphi$ for RE $n$ in block $t$ \\
\hline

$g_n^t(\varphi)$ & stochastic gradient for updating
$H_n^t(\varphi)$ in block $t$ \\
\hline

$\gamma^t$ & learning rate in block $t$ \\
\hline

\end{tabular}
\end{table}

Our work is most closely related to the blind beamforming method in \cite{Arun2020RFocus, ren2022configuring} based on statistics. Specifically, \cite{ren2022configuring} initiates the idea of using the conditional sample mean (CSM) to configure the IRS for an isolated point-to-point link. Its main idea follows: Choose phase shifts for all REs at random, then use the average performance conditioned on a particular choice of phase shift for the $n$th RE to characterize the goodness of this phase shift decision, and finally choose a phase shift for each RE $n$ to maximize the corresponding conditional average performance. More recently, the CSM algorithm has been extended for multiple IRSs \cite{Xu2024Coordinating}, for multiple users \cite{xu2024blind}, and for the Non-line-of-sight (NLoS) transmission \cite{lai2025adaptive}. While \cite{Arun2020RFocus, ren2022configuring} give the statistical motivation behind blind beamforming, this work shows that blind beamforming can be interpreted as an $\epsilon$-greedy algorithm \cite{watkins1989learning} for solving the multi-armed bandit problem in the reinforcement learning field; the key step is to figure out ``cumulative regret'' in our problem case. Following this line, we further extend the gradient bandit method \cite{sutton2018reinforcement} to improve the efficiency of blind beamforming. Intuitively, rather than trying out new samples uniformly, we propose trying out those more ``promising'' samples with higher priorities so as to attain the optimum more quickly. More importantly, we show that the proposed gradient sampling scheme is a stochastic approximation to the gradient ascent, so its convergence is guaranteed. Furthermore, we build a prototype system to evaluate this gradient sampling scheme in a real-world testbed; the field test results show that it can outperform the existing blind beamforming method significantly.


\textit{Notation:}
Here and throughout, the set of complex numbers is denoted as $\mathbb{C}$. Let $\mathbb E[\cdot]$ denote expectation. The indicator function is denoted by
$\mathbbm 1\{\cdot\}$. The notation
$\operatorname{Uniform}(\mathcal A)$ represents the uniform
distribution over a finite set $\mathcal A$.  A summary of main variables is provided in Table \ref{tab:main_variables} for ease of reference.

\section{System Model}
\label{sec:sys_model}

Consider an IRS-assisted wireless network as shown in Fig. \ref{fig:system}. Assume that the transmitter and receiver have one antenna each, and that the IRS consists of $N$ REs in total. Every RE can induce a reflected path from the transmitter to the receiver. Denote by $h_0\in\mathbb C$ the direct channel from the transmitter to the receiver, $h_{n,t}\in\mathbb C$ the channel from the transmitter to RE $n$, and $h_{n,r}\in\mathbb C$ the channel from RE $n$ to the receiver, for $n=1,\ldots,N$.
Each cascaded reflected channel $h_n$ can then be factored as
\begin{equation}
    h_n = h_{n,r}\times h_{n,t}.
\end{equation}
Note that we do not assume any channel model for both the direct and reflected channels. In this work, we assume that the channels are unknown \emph{a priori}.

\begin{figure}[t]
    \centering
    \includegraphics[width=0.9\linewidth]{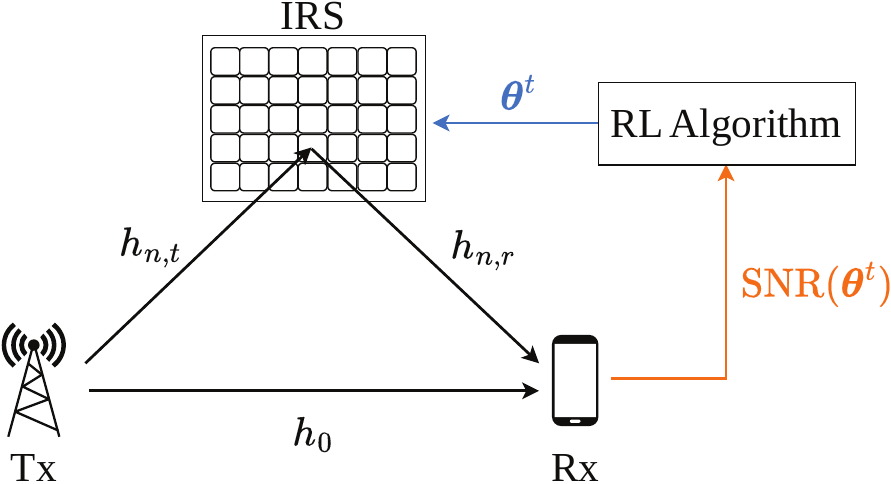}
    \caption{The IRS-assisted wireless network.}
    \label{fig:system}
\end{figure}

Moreover, each RE $n$ can induce a phase shift in the corresponding reflected channel. Assume that the configuration of IRS can be updated on a per-block basis; let $\theta^t_n\in [0, 2\pi)$ be the phase shift of RE $n$ in block $t$. Assume also that there are a total of $T$ blocks and the channels are invariant throughout. In practice, the value of each phase shift is restricted to a prescribed discrete set
\begin{equation}
\Phi_K = \{0,\omega,\ldots,(K-1)\omega\}\;\;\text{with}\;\;\omega = \frac{2\pi}{K},
\end{equation}
where the set size $K\ge2$.
With the background noise $Z^t\sim\mathcal{CN}(0,\sigma^2)$ and the transmit signal $X^t\sim\mathcal{CN}(0,P)$, the received signal $Y^t\in\mathbb C$ in block $t$ is
\begin{equation}
Y^t = \left(h_0+\sum_{n=1}^{N}h_ne^{j\theta_n^t}\right)X^t+Z^t,
\end{equation}
and the SNR in block $t$ can be computed as a function of the corresponding phase shift array $\bm\theta^{t} = (\theta^{t}_1,\ldots,\theta^{t}_N)$ as
\begin{equation}
\label{SNR}
\mathrm{SNR}(\bm\theta^t) = \frac{P\left|h_0+\sum_{n=1}^{N}h_ne^{j\theta^{t}_{n}}\right|^2}{\sigma^2}.
\end{equation}
Roughly speaking, our goal is to maximize SNR for all the $T$ blocks, which will be mathematically formalized in the context of reinforcement learning in the following section. Again, we emphasize that the channels are unknown \emph{a priori}.

\section{A Reinforcement Learning Perspective of Existing Blind Beamforming Methods}
\label{sec:RL_perspective}

\subsection{Review of Existing Blind Beamforming \cite{Arun2020RFocus, ren2022configuring}}
Should we alter the phase shifts when entering a new block? This question becomes trivial if CSI is already known. Since the channels are constant, we can just optimize the phase shift by solving the SNR maximization problem
\begin{subequations}
\label{eq:SNR_problem}
\begin{align}
    &\underset{\bm\theta}{\text{maximize}} \quad \frac{P\bigg|h_0+\sum_{n=1}^{N}h_ne^{j\theta_n}\bigg|^2}{\sigma^2}\\
    &\text {subject to} \quad\, \theta_n \in \Phi_K,\quad n=1, \ldots, N,
\end{align}
\end{subequations}
for the first block and fix the solution thereafter. Regarding the SNR maximization with CSI available, a common idea 
\cite{wu2020beamforming,you2020channel,abeywickrama2020intelligent,wang2020intelligent} is to rotate every $h_n$ to the closest position to $h_0$, i.e., 
\begin{equation}
\label{CPP}
\theta^\star_{n} = \min_{\varphi\in\Phi_K}\left|\mathsf{Arg} \left(\frac{h_ne^{j\varphi}}{h_0}\right)\right|,
\end{equation}
where $\mathsf{Arg}$ is the principal argument, namely the closest point projection (CPP).

\begin{algorithm}[t]
\caption{CSM blind beamforming method in \cite{Arun2020RFocus, ren2022configuring}}
\label{alg:CSM}
\begin{algorithmic}[1]

\State \textbf{Input:} Discrete set $\Phi_K$, number of blocks $T$, number of exploration blocks $T_0$
\For{$t=1,\ldots,T_0$}
    \State Sample $\theta_n^t\sim \mathrm{Uniform}(\Phi_K)$, $\forall n=1,\ldots,N$

    \State Apply $\boldsymbol{\theta}^t=(\theta_1^t,\ldots,\theta_N^t)$ and measure
    $\mathrm{SNR}(\boldsymbol{\theta}^t)$.
\EndFor
\State Compute the CSM $Q_n^{T_0}(\varphi)$ according to \eqref{eq:compute_Q}.
\State Decide each $\theta_n$ according to \eqref{eq:CSM_solution}.
\For{$t=T_0+1,\ldots,T$}

    \State Apply $\boldsymbol{\theta}^t=(\theta_1^{\mathrm{CSM}},\ldots,\theta_N^{\mathrm{CSM}})$.
\EndFor
\end{algorithmic}
\end{algorithm}

But the reality is that the CSI is unknown. It turns out that the CPP solution can be recovered implicitly via random sampling. Specifically, as proposed in \cite{Arun2020RFocus, ren2022configuring}, we try out a new sample $\bm\theta^t$ with each $\theta^t_n$ drawn from $\Phi_K$ i.i.d. and then record the corresponding $\mathrm{SNR}(\bm\theta^t)$. After $T_0$ samples, we compute the conditional sample mean (CSM)
\begin{equation}
    \label{eq:compute_Q}
    Q^{T_0}_{n}(\varphi) = \frac{\sum^{T_0}_{i=1}\mathrm{SNR}(\bm\theta^{i})\cdot\mathbbm 1_{\theta^{i}_n=\varphi}}{\sum^{T_0}_{i=1}\mathbbm 1_{\theta^{i}_n=\varphi}}
\end{equation}
for each $\varphi\in\Phi_K$. The blind beamforming method in \cite{Arun2020RFocus, ren2022configuring} is to choose each phase shift $\theta_n$ to maximize the CSM, i.e.,
\begin{equation}
    \label{eq:CSM_solution}
    \theta_n^{\mathrm{CSM}} = \arg\mathop{\max}\limits_{\varphi\in\Phi_K}Q^{T_0}_n(\varphi).
\end{equation}
The main result in \cite{Arun2020RFocus, ren2022configuring} follows:
\begin{proposition}[Theorem 2 in \cite{ren2022configuring}]
    As $T_0\rightarrow\infty$, the solution $\theta_n^{\mathrm{CSM}}$ in \eqref{eq:CSM_solution} tends to the CPP solution $\theta^\star_n$ in \eqref{CPP}.
\end{proposition}

But we cannot let $T_0\rightarrow\infty$ since the total number of blocks $T$ is finite. A practical realization of blind beamforming in our case is to use $T_0<T$ blocks to learn the CSM values $\{Q^{T_0}_n(\varphi)\}$, namely training. For the rest blocks, we decide $\bm\theta^t$ by the blind beamforming method \cite{Arun2020RFocus, ren2022configuring} with the estimated $\{Q^{T_0}_n(\varphi)\}$. Thus, the above blind beamforming method can be recognized as a reinforcement learning approach: \emph{explore} new solutions for the first $T_0$ blocks and \emph{exploit} the CSM so far for the remaining $T-T_0$ blocks.

Furthermore, our problem can be thought of as a multi-armed bandit problem with each possible solution $\bm\theta^t\in\Phi_K^N$ treated as an arm and the received SNR under $\bm\theta^t$ is treated as the reward.
Following this line, we can generalize the CSM blind beamforming method \cite{Arun2020RFocus, ren2022configuring} as an $\epsilon$-greedy algorithm---which is a classical approach to the multi-armed bandit problem. At the beginning of each block $t$, we determine whether the block is used for exploration or exploitation. Specifically, block $t$ is selected as an exploration block with probability $0<\epsilon<1$ and as an exploitation block with probability $1-\epsilon$. Therefore, for each block $t$, the phase shift $\theta_n$ is determined as
\begin{equation}\label{greedy}
\theta_n^{t}=\left\{
\begin{aligned}
&\arg\mathop{\max}\limits_{\varphi\in\Phi_K}Q^t_n(\varphi)\; \text{with probability}\; 1-\epsilon;\\
&\varphi \sim \mathrm{Uniform} (\Phi_K)\;\text{with probability}\; \epsilon.
\end{aligned}
\right.
\end{equation}
Note that the $Q_n^t(\varphi)$ is now computed as the average SNR over all exploration blocks so far, i.e.,
\begin{equation}
\label{eq:update_Q}
    Q_n^t(\varphi)
=
\frac{
\sum_{\tau=1}^{t-1}
\mathrm{SNR}(\boldsymbol{\theta}^{\tau})
\cdot \mathbbm{1}_{\theta_n^\tau=\varphi}
\cdot \mathbbm{1}_{\tau \in \mathcal{T}_t}
}{
\sum_{\tau=1}^{t-1}
\mathbbm{1}_{\theta_n^\tau=\varphi}
\cdot \mathbbm{1}_{\tau \in \mathcal{T}_t}
},
\end{equation}
where $\mathcal{T}_t$ is the set of indices of all exploration blocks before block $t$.
The above process is summarized in Algorithm \ref{alg:epsilon_greedy}. In contrast to the CSM blind beamforming method in \cite{Arun2020RFocus, ren2022configuring} that uses the first $T_0$ blocks for exploration, we now decide the use of each block $t$, exploration or exploitation, randomly.

The above interpretation further leads us to the definition of \emph{regret} as considered in the reinforcement learning field. Since 
the CPP solution $\{\theta^\star_n\}$ is the ultimate goal and it can be obtained from $Q^{T_0}_{n}(\varphi)$ with $T_0\rightarrow\infty$, we may consider the following regret for the action $\bm\theta^t$ in block $t$:
\begin{equation}
    G_t = \sum^N_{n=1} \left(Q^\infty_n(\theta_n^\star)-Q^\infty_n(\theta_n^{t})\right).
\end{equation}
It can be shown that $G_t$ must be nonnegative and it equals 0 if $\theta_n^t = \theta^\star_n$ for every $n$. We seek the optimal $\{\bm\theta^{t}\}$ to minimize the expected cumulative regret over the $T$ blocks:
\begin{subequations}
\label{new opt problem}
\begin{align}
    &\underset{\{\bm{\theta}^t\}}{\text{minimize}} \quad \mathbb{E}\left[\sum_{t=1}^{T}G_t\right]
    \label{new opt problem:obj}\\
    &\text {subject to} \quad \theta_{n}^t \in \Phi_K,\;\text{for each}\;\;(n,t).
    \label{new opt problem:contraint}
\end{align}
\end{subequations}
The $\epsilon$-greedy algorithm in \eqref{greedy} constitutes the standard approach to the above problem in the reinforcement learning field. But we raise two questions: What is the optimal value of $\epsilon$ in \eqref{greedy}? Is there a better way of exploration than the uniform sampling? We will answer the first question in the following subsection and defer the second question to Section \ref{sec:gradient}.

\begin{algorithm}[t]
\caption{$\epsilon$-Greedy Algorithm}
\label{alg:epsilon_greedy}
\begin{algorithmic}[1]

\State \textbf{Input:} Discrete set $\Phi_K$, number of blocks $T$, exploration probability $\epsilon$
\State Initialize $Q_n^1(\varphi)=0$
\For{$t=1,\ldots,T$}
    \State Update $Q_n^t(\varphi)$ according to \eqref{eq:update_Q}.
    \State Decide block $t$ according to the probability $\epsilon$.

    \If{block $t$ is an exploration block}
        \State Sample $\theta_n^t\sim \mathrm{Uniform}(\Phi_K)$, $\forall n=1,\ldots,N$
    \Else
        \State Set $\theta_n^t\in\arg\max_{\varphi\in\Phi_K} Q_n^t(\varphi), \forall n=1,\ldots,N$
    \EndIf

    \State Apply $\boldsymbol{\theta}^t=(\theta_1^t,\ldots,\theta_N^t)$ and measure
    $\mathrm{SNR}(\boldsymbol{\theta}^t)$.
\EndFor
\end{algorithmic}
\end{algorithm}

\subsection{Parameter Tuning for $\epsilon$-Greedy Algorithm}
\label{subsec:greedy}

Clearly, there is a trade-off in deciding the value of $0<\epsilon<1$. If $\epsilon\rightarrow0$, then the distortion between the real CSM $Q^\infty_n(\varphi)$ and the estimated CSM  $Q^t_n(\varphi)$ can be large, so the resulting blind beamforming solution is of low quality. Conversely, if $\epsilon\rightarrow1$, then most blocks are used for exploration and consequently SNR should be fairly low for most of the time.
But optimizing $\epsilon$ directly in \eqref{new opt problem} is quite difficult because the objective is difficult to express in a closed form with respect to $\epsilon$. To solve this issue, the following proposition provides an upper approximation of regret that facilitates the optimization significantly.

\begin{proposition}
\label{prop:greedy_bound}
For the $\epsilon$-greedy method, the expected value of the cumulative regret over $T$ blocks is upper bounded as
\begin{equation}
\label{eq:upper_bound}
\mathbb{E}\left[\sum_{t=1}^{T}G_t\right] \leq\rho TN\left(\frac{2K}{T^4}+\epsilon\right)+2N\rho(1-\epsilon )\sqrt{\frac{2KT\log T}{\epsilon}}
\end{equation}
given any parameter $0<\epsilon<1$, where $\rho = \frac{P}{\sigma^2}\left(\sum_{n=0}^{N}|h_n|\right)^2$.
\end{proposition}
\begin{IEEEproof}
We begin with an exploitation block $t$. Defining an event for each $\varphi\in\Phi_K$ as
\begin{equation}
    \label{eq:event}
    \mathcal{E}(\varphi)=\left\{\left|Q_n^t(\varphi)-Q_n^\infty(\varphi)\right|> \rho\sqrt{\frac{2K\log T}{\epsilon\cdot T}}\right\},
\end{equation}
we show that
\begin{align}
\label{eq:pro_event}
&\mathbb{P}\left\{\left|Q_n^t(\varphi)-Q_n^\infty(\varphi)\right|\leq \rho\sqrt{\frac{2K\log T}{\epsilon\cdot T}},\;\forall \varphi\in\Phi_K\right\}\notag\\ 
&= 1-\mathbb{P}\Big\{\bigcup\nolimits_{\varphi\in\Phi_K}\mathcal{E}(\varphi)\Big\}\notag\\
&\stackrel{(a)}{\geq}1-\sum_{\varphi\in\Phi_K}\mathbb{P}\left\{\mathcal E(\varphi)\right\} \notag \\
&\stackrel{(b)}{\geq} 1-\frac{2K}{T^4},
\end{align}
where $(a)$ follows by the union bound while $(b)$ follows by Hoeffding's inequality. Moreover, the exploitation in \eqref{greedy} guarantees that
\begin{equation}
 \label{eq:regret_cond}
Q_n^t(\theta_n^t)\geq Q_n^t(\theta_n^\star).  
\end{equation}
Combining \eqref{eq:pro_event} and \eqref{eq:regret_cond}, we have
\begin{equation}
\mathbb{P}\left\{Q_n^\infty(\theta^\star_n)-Q_n^\infty(\theta^t_n)\le 2\rho\sqrt{\frac{2K\log T}{\epsilon\cdot T}}\right\}\ge 1-\frac{2K}{T^4}.
\end{equation}

We now consider a general block $t=1,\ldots,T$. Clearly, it always holds that
\begin{equation}
\label{eq:bound_1}
Q^\infty_n(\theta_n^\star)-Q^\infty_n(\theta_n^{t})
\leq \rho.
\end{equation}
Thus, the expected value of $Q^\infty_n(\theta_n^\star)-Q^\infty_n(\theta_n^{t})$ can be bounded from above as
\begin{align}
&\mathbb{E}\left[Q^\infty_n(\theta_n^\star)-Q^\infty_n(\theta_n^{t})\right] \notag \\
    &\leq \epsilon\rho+(1-\epsilon)\left(\left(1-\frac{2K}{T^4}\right)2\rho\sqrt{\frac{2K\log T}{\epsilon\cdot T}}+\frac{2K}{T^4}\rho\right) \notag\\
    &\leq \epsilon\rho+2(1-\epsilon)\rho\sqrt{\frac{2K\log T}{\epsilon\cdot T}} + \frac{2K}{T^4}\rho.
\label{Q:diff}
\end{align}
Summing both sides of \eqref{Q:diff} gives rise to \eqref{eq:upper_bound}.
\end{IEEEproof}

We now use the upper approximation in \eqref{eq:upper_bound} to replace the objective function in \eqref{new opt problem}. The resulting new problem turns out to be convex in $\epsilon$, so it suffices to solve the first-order condition:
\begin{equation}
\label{first order cond}
    \sqrt{\epsilon^3 T^3}-T(1+\epsilon )\sqrt{2K\log T}=0.
\end{equation}
Further, it can be shown that the above first-order equation has only one real root 
\begin{equation}
\label{opt:epsilon}    \epsilon^\star=\sqrt{\frac{2K\log T}{9T}}+\sqrt[3]{a+\sqrt{a^2+b^3}} 
    +\sqrt[3]{a-\sqrt{a^2+b^3}},
\end{equation}
where 
\begin{subequations}
\begin{align}
    a&=\left(1+\frac{4K\log T}{27T}\right)\sqrt{\frac{K\log T}{2T}} \\
    b&=\frac{2K\log T}{9T}.
\end{align}
\end{subequations}

For the CSM method, if $T_0$ is relaxed to a continuous variable, then we have $T_0^\star=\epsilon^\star T$, since it is a special case of the $\epsilon$-greedy algorithm. However, $T_0$ must be an integer, and thus $T_0^\star$ can be obtained by rounding $\epsilon^\star\cdot T$ to the nearest integer.

\section{Proposed Non-Uniform Sampling}
\label{sec:gradient}

\subsection{Gradient Sampling}
Our discussion thus far is based on uniform sampling---which is questionable since it does not take the past performance into account when choosing a new $\bm\theta^t$ for exploration. For instance, if using a particular phase shift $\varphi\in\Phi_K$ for RE $n$ always results in poor performance, then it is advisable to avoid this setting in future explorations.

As the main contribution of this work, we extend the gradient bandit method in \cite{sutton2018reinforcement} to address the above issue, as stated in what follows. Specifically, we introduce a \emph{numerical preference} variable $H_n^t(\varphi)$ for each $(n,\varphi)$ in block $t$ to characterize the ``goodness'' of setting $\theta^t_n=\varphi$. Roughly speaking, for each RE $n$, if a phase shift option $\varphi \in \Phi_K$ has a larger numerical preference variable $H_n^t(\varphi)$ at block $t$, it will have a higher probability of being sampled in this block. To achieve this, we incorporate the softmax probability
\begin{equation}
\label{eq:prob_phi}
\pi_n^t(\varphi)\triangleq\mathbb{P}\left\{\theta_n^t=\varphi\right\}\triangleq\frac{e^{H_n^t(\varphi)}}{\sum_{\varphi^\prime\in\Phi_K}e^{H_n^t(\varphi^\prime)}},
\end{equation}
and make a soft decision for each $\theta^t_n$ as
\begin{equation}
    \label{eq:distribution}
    \theta_n^t=\varphi\;\text{with probability}\; \pi^t_n(\varphi).
\end{equation}

\begin{algorithm}[t]
\caption{Gradient Sampling Algorithm}
\label{alg:gradient_sampling}
\begin{algorithmic}[1]

\State \textbf{Input:} set $\Phi_K$, number of blocks $T$, learning rates $\{\gamma^t\}$
\State Initialize $H_n^1(\varphi)=0,\forall n,\varphi$.

\For{$t=1,\ldots,T$}

    \For{$n=1,\ldots,N$}
        \State Compute
        $\pi_n^t(\varphi),\,\forall \varphi\in\Phi_K$ according to \eqref{eq:prob_phi}.

        \State Decide each $\theta^t_n$ according to \eqref{eq:distribution}.
    \EndFor

    \State Apply $\boldsymbol{\theta}^t=(\theta_1^t,\ldots,\theta_N^t)$ and measure
    $\mathrm{SNR}(\boldsymbol{\theta}^t)$.
    
    \State Compute $\delta^t$ according to \eqref{eq:delta_t}.
    
    \State Obtain $H_n^{t+1}(\varphi),\forall n,\varphi$ according to \eqref{eq:update_H}.

\EndFor
\end{algorithmic}
\end{algorithm}

It remains to specify how each $H^t_n(\varphi)$ is obtained. By convention, every $H^1_n(\varphi)$ is set to zero. Sequentially, each $H^{t+1}_n(\varphi)$ is computed based on the previous $\{H_n^{t}(\cdot),\pi^{t}_n(\varphi)\}$ and the past rewards as
\begin{equation}
    \label{eq:update_H}
    H_n^{t+1}(\varphi) = H_n^{t}(\varphi) + \gamma^t\cdot g_n^t(\varphi),
\end{equation}
where
\begin{align}
    \label{eq:g_nt}
    g_n^t(\varphi) &= \delta^t\left(\mathbbm{1}\{\varphi = \theta_n^t\}-\pi^t_n(\varphi)\right) \notag \\
    &=\begin{cases}
  \delta^t(1-\pi_n^{t}(\theta_n^t)), & \text{if } \varphi=\theta_n^t, \\
  -\delta^t\pi_n^{t}(\varphi),  & \forall\varphi\ne\theta_n^t,
\end{cases}
\end{align}
and $\gamma^t>0$ is the learning rate in block $t$. Moreover,
\begin{equation}
    \label{eq:delta_t}
    \delta^t = \mathrm{SNR}(\bm\theta^t) - \overline{r}^{t-1}
\end{equation}
is the difference between the latest reward (which is the SNR value) and the average reward before block $t$:
\begin{equation}
\bar r^{t-1}
\triangleq
\frac{1}{t-1}\sum_{i=1}^{t-1}\mathrm{SNR}(\boldsymbol{\theta}^i).
\end{equation}
By convention, we set $\bar r^{0}=0$. According to \eqref{eq:update_H}, if the SNR under the current sample is higher than the average SNR up to the previous round, we increase the numerical preference of the phase selected for each RE $n$. The process of the proposed gradient sampling algorithm is summarized in Algorithm \ref{alg:gradient_sampling}. Note that we no longer distinguish between exploration and exploitation when applying the gradient sampling algorithm. In the following subsection, we show that the proposed gradient sampling algorithm is an instance of stochastic gradient ascent, thereby guaranteeing its convergence.

\begin{remark}
Although the proposed gradient sampling algorithm is originally designed to maximize the received SNR, it can be readily extended to other objectives. Specifically, let $f(\bm{\theta})$ denote an objective function defined with respect to the IRS configuration $\bm{\theta}$. For a maximization problem, the gradient sampling algorithm can be used by replacing $\mathrm{SNR}(\bm{\theta}^t)$ with $f(\bm{\theta}^t)$ as the reward in block $t$. For a minimization problem, $-f(\bm{\theta}^t)$ can be used as the reward in block $t$. Therefore, the proposed gradient sampling algorithm is applicable to tasks such as SNR reduction and multi-user max--min fairness. We also evaluate the performance of gradient sampling on these tasks in Section \ref{sec:field}.
\end{remark}

\subsection{Convergence Analysis}

This section aims to provide a deeper insight into the gradient sampling algorithm by extending the analysis in \cite{sutton2018reinforcement}. Specifically, we will show that the gradient sampling algorithm is a stochastic approximation of gradient ascent to solve the expected SNR maximization problem, and therefore we can characterize its convergence behavior. 

Let
\begin{equation}
    \mathbf H \triangleq \{H_n(\varphi): n=1,\ldots,N,\ \varphi\in\Phi_K\}
\end{equation}
be the set of all numerical preference variables. We now suppress the block index and write each $H_n(\phi)$ as a generic numerical preference variable; its value at block $t$ generated by Algorithm \ref{alg:gradient_sampling} is denoted by $H_n^t(\phi)$. Given $\mathbf H$, the REs independently sample their phase shifts according to the softmax probabilities in \eqref{eq:prob_phi}. Therefore, the probability of a specific phase-shift vector $\boldsymbol{\theta}=(\theta_1,\theta_2,\ldots,\theta_N)$ is given by
\begin{equation}
\prod_{n=1}^N \pi_n(\theta_n)= \prod_{n=1}^N\frac{e^{H_n(\theta_n)}}{\sum_{\varphi^\prime\in\Phi_K}e^{H_n(\varphi^\prime)}}.
\end{equation}
By treating $\mathbf H$ as the optimization variable, we define the expected SNR as a function of $\mathbf H$:
\begin{align}
R(\mathbf H)
&\triangleq
\mathbb E_{\boldsymbol{\theta}}
\left[\mathrm{SNR}(\boldsymbol{\theta})\right] \notag \\
&=\sum_{\boldsymbol{\theta}\in\Phi_K^N}
\left(
\prod_{n=1}^N \pi_n(\theta_n)
\right)
\mathrm{SNR}(\boldsymbol{\theta}).
\label{eq:expected_snr_objective}
\end{align}
Rather than directly solving the SNR maximization problem \eqref{eq:SNR_problem}, the proposed gradient sampling algorithm aims to solve the expected SNR maximization problem:
\begin{equation}
    \underset{\mathbf H}{\text{maximize}}\quad R(\mathbf H).
    \label{eq:actual_problem}
\end{equation}
In other words, the gradient sampling algorithm aims to maximize the expected SNR by optimizing $\bm H$ to adjust the sampling distribution of $\bm{\theta}$, such that phase-shift vectors associated with higher SNR values are sampled more frequently. We remark that the above problem reformulation is reasonable because problem \eqref{eq:actual_problem} is equivalent to the original problem \eqref{eq:SNR_problem} based on the following observation: if $\bm\theta=(\theta_1^\star,\theta_2^\star,\ldots,\theta_N^\star)$ is an optimal phase-shift vector, then we can assign $\mathbb{P}\left\{\theta_n=\theta_n^\star\right\} = 1$ by setting $H_n(\theta_n^\star)=+\infty$.

We next show that the proposed gradient sampling algorithm is a stochastic approximation of gradient ascent to solve problem \eqref{eq:actual_problem}. It suffices to show that $g_n^t(\varphi)$ in \eqref{eq:g_nt} is an unbiased estimator of the gradient of $R(\bH)$ with respect to the variable $H_n(\varphi)$ in block $t$, as summarized in the following proposition.

\begin{proposition}
For every block $t\geq 1$, RE $n$, and phase shift $\varphi\in\Phi_K$, the gradient sampling update $g_n^t(\varphi)$ in \eqref{eq:g_nt} is an unbiased estimator of the partial derivative of the expected SNR objective \eqref{eq:expected_snr_objective} with respect to the numerical preference variable $H_n(\varphi)$, i.e.,
\begin{equation}
\mathbb E_{\boldsymbol\theta^t}\left[g_n^t(\varphi)\right]
=\left. \frac{\partial R(\mathbf H)}{\partial H_n(\varphi)} \right|_{\mathbf H=\mathbf H^t},
\label{eq:unbiased_coordinate_gradient}
\end{equation}
where $\boldsymbol\theta^t$ is the random phase-shift vector given by the gradient sampling algorithm in block $t$.
\label{prop:unbiased}
\end{proposition}

\begin{IEEEproof}
We first derive the partial derivative of $R(\mathbf H)$ with respect to $H_n(\varphi)$ for an arbitrary RE $n$ and phase shift $\varphi\in\Phi_K$. By the definition of $R(\mathbf H)$ in \eqref{eq:expected_snr_objective}, when the partial derivative is taken with respect to $H_n(\varphi)$, only the $n$th softmax term depends on $H_n(\varphi)$. Therefore,
\begin{align}
\label{eq:partial_H}
\frac{\partial R(\mathbf H)}
{\partial H_n(\varphi)}
&=
\sum_{\boldsymbol{\theta}\in\Phi_K^N}
\mathrm{SNR}(\boldsymbol{\theta})
\left(
\prod_{m\neq n}\pi_m(\theta_m)
\right)
\frac{\partial \pi_n(\theta_n)}
{\partial H_n(\varphi)}.
\end{align}
Based on the softmax probability given in \eqref{eq:prob_phi}, we have
\begin{align}
\label{eq:partial_softmax}
&\frac{\partial \pi_n(\theta_n)}{\partial H_n(\varphi)} \notag\\
&=
\frac{\partial}{\partial H_n(\varphi)}
\left[
\frac{e^{H_n(\theta_n)}}{\sum\limits_{\theta_n^\prime\in\Phi_K} e^{H_n(\theta_n^\prime)}}
\right] \notag\\
&=
\frac{
\frac{\partial e^{H_n(\theta_n)}}{\partial H_n(\varphi)}
\sum_{\theta_n^\prime\in\Phi_K} e^{H_n(\theta_n^\prime)}
-
e^{H_n(\theta_n)}
\frac{\partial \sum_{\theta_n^\prime\in\Phi_K} e^{H_n(\theta_n^\prime)}}{\partial H_n(\varphi)}
}{
\left(\sum_{\theta_n^\prime\in\Phi_K} e^{H_n(\theta_n^\prime)}\right)^2
} \notag\\
&=
\frac{
\mathbbm{1}\{\varphi=\theta_n\} e^{H_n(\theta_n)}
}{
\sum_{\theta_n^\prime\in\Phi_K} e^{H_n(\theta_n^\prime)}
}
-
\frac{
e^{H_n(\theta_n)} e^{H_n(\varphi)}
}{
\left(\sum_{\theta_n^\prime\in\Phi_K} e^{H_n(\theta_n^\prime)}\right)^2
}
\notag\\
&=
\pi_n(\theta_n)\left(\mathbbm{1}\{\varphi=\theta_n\}-\pi_n(\varphi)\right).
\end{align}
Substituting \eqref{eq:partial_softmax} into \eqref{eq:partial_H} gives
\begin{equation}
\frac{\partial R(\mathbf H)}
{\partial H_n(\varphi)}
=
\mathbb E_{\boldsymbol\theta}\left[
\mathrm{SNR}(\boldsymbol{\theta})
\left(
\mathbbm{1}\{\varphi=\theta_n\}
-
\pi_n(\varphi)
\right)
\right],
\label{eq:gradient_coordinate}
\end{equation}
and therefore
\begin{equation}
    \left. \frac{\partial R(\mathbf H)}{\partial H_n(\varphi)} \right|_{\mathbf H=\mathbf H^t}=\mathbb E_{\boldsymbol\theta^t}\left[
\mathrm{SNR}(\boldsymbol{\theta}^t)
\left(
\mathbbm{1}\{\varphi=\theta_n^t\}
-
\pi_n(\varphi)
\right)
\right].
\end{equation}

We then compute the expectation of $g_n^t(\varphi)$:
\begin{subequations}
\label{eq:baseline_decomposition}
\begin{align}
&\mathbb E_{\boldsymbol\theta^t}\left[g_n^t(\varphi)\right] \notag \\
&=
\mathbb E_{\boldsymbol\theta^t}\left[
\left(
\mathrm{SNR}(\boldsymbol{\theta}^t)-\bar r^{t-1}
\right)
\left(
\mathbbm{1}\{\varphi=\theta_n^t\}
-
\pi_n^t(\varphi)
\right)
\right]\label{eq:E_gnt_1}\\
&=
\mathbb E_{\boldsymbol\theta^t}\left[
\mathrm{SNR}(\boldsymbol{\theta}^t)
\left(
\mathbbm{1}\{\varphi=\theta_n^t\}
-
\pi_n^t(\varphi)
\right)
\right] \notag \\
&\quad
-
\bar r^{t-1}
\mathbb E_{\boldsymbol\theta^t}\left[
\mathbbm{1}\{\varphi=\theta_n^t\}
-
\pi_n^t(\varphi)
\right].\label{eq:E_gnt_2}
\end{align}
\end{subequations}
Note that we can derive from \eqref{eq:E_gnt_1} to \eqref{eq:E_gnt_2} because $\bar r^{t-1}$ is a constant for the current block $t$.
Furthermore, the probability of $\varphi=\theta_n^t$ is given by $\pi_n^t(\varphi)$. Hence,
\begin{equation}
\mathbb E_{\boldsymbol\theta^t}\left[
\mathbbm{1}\{\varphi=\theta_n^t\}
-
\pi_n^t(\varphi)
\right]
=
\pi_n^t(\varphi)-\pi_n^t(\varphi)
=
0.
\end{equation}
Thus, \eqref{eq:baseline_decomposition} reduces to
\begin{equation}
\mathbb E_{\boldsymbol\theta^t}\left[g_n^t(\varphi)\right]
=
\mathbb E_{\boldsymbol\theta^t}\left[
\mathrm{SNR}(\boldsymbol{\theta}^t)
\left(
\mathbbm{1}\{\varphi=\theta_n^t\}
-
\pi_n^t(\varphi)
\right)
\right].
\label{eq:expect_gnt}
\end{equation}
Comparing \eqref{eq:expect_gnt} with \eqref{eq:gradient_coordinate} proves
\eqref{eq:unbiased_coordinate_gradient}.
\end{IEEEproof}

Proposition \ref{prop:unbiased} reveals that the expected update of the gradient sampling algorithm equals the gradient of the expected SNR with respect to the numerical preference variable. Therefore, the algorithm is an instance of stochastic gradient ascent for solving problem \eqref{eq:actual_problem}. This indicates that the gradient sampling algorithm has good convergence properties, as shown in the following two theorems.

\begin{theorem}[Asymptotic first-order stationarity]
\label{theorem:asymptotic}
Consider the gradient sampling update in \eqref{eq:update_H} where $g_n^t(\varphi)$ is defined in \eqref{eq:g_nt}. If the learning rate $\gamma^t$ satisfies the Robbins–Monro conditions, i.e.,
\begin{equation}
\label{eq:gamma_condition}
    \gamma^t > 0, \quad
\sum_{t=2}^{\infty} \gamma^t = \infty, \quad
\sum_{t=2}^{\infty} (\gamma^t)^2 < \infty,
\end{equation}
then the generated sequence is asymptotically first-order stationary, i.e., for all $n$ and $\varphi$, we have
\begin{equation}
    \lim_{t\to\infty}\mathbb E_{\boldsymbol\theta^t}\left[g_n^t(\varphi)\right] = \lim_{t\to\infty}\left. \frac{\partial R(\mathbf H)}{\partial H_n(\varphi)} \right|_{\mathbf H=\mathbf H^t}=0 \text{ almost surely.}
\end{equation}
\end{theorem}

\begin{IEEEproof}
We prove the theorem by verifying the standard conditions for stochastic gradient convergence in \cite{bertsekas2000gradient}. Specifically, a stochastic gradient ascent recursion is asymptotically first-order stationary if the following four conditions hold:
\begin{enumerate}
\item the learning rates satisfy the Robbins--Monro conditions;
\item the stochastic gradient is unbiased;
\item the stochastic gradient has a bounded second moment;
\item the objective function is globally smooth and bounded.
\end{enumerate}
The second condition has been established in Proposition \ref{prop:unbiased}, and the first condition is assumed. Therefore, it remains to prove the last two conditions.

We first verify the third condition. By the definition of $g_n^t(\varphi)$ in \eqref{eq:g_nt}, we have
\begin{equation}
    |\delta^t| \leq \max_{\bm\theta\in\Phi_K^N}\mathrm{SNR}(\bm\theta),
\end{equation}
and
\begin{equation}
    |\mathbbm{1}\{\varphi = \theta_n^t\}-\pi^t_n(\varphi)| \leq 1.
\end{equation}
Therefore,
\begin{equation}
|g_n^t(\varphi)|^2\leq\left(\max_{\bm\theta\in\Phi_K^N}\mathrm{SNR}(\bm\theta)\right)^2<+\infty,
\end{equation}
so the stochastic gradient has a bounded second moment.

It remains to verify the last condition. We first show that $R(\mathbf H)$ is globally smooth. Since the action space $\Phi_K^N$ is finite and each softmax probability $\pi_n(\theta_n)$ is a smooth function of the corresponding preference variables ${H_n}$ with bounded derivatives, $R(\mathbf H)$ is a finite sum of smooth functions with bounded derivatives. Hence, $R(\mathbf H)$ is globally smooth. We then show that $R(\mathbf H)$ is bounded. Clearly, we have
\begin{equation}
0\leq R(\mathbf H)
\leq
\max_{\bm\theta\in\Phi_K^N}\,\mathrm{SNR}(\bm\theta).
\label{eq:R_bounded_above}
\end{equation}

Thus, the proposed gradient sampling algorithm satisfies all four conditions. By Proposition 3 in \cite{bertsekas2000gradient}, the sequence generated by \eqref{eq:update_H} is asymptotically first-order stationary.
\end{IEEEproof}

\begin{figure}[t]
    \centering
    \includegraphics[width=\linewidth]{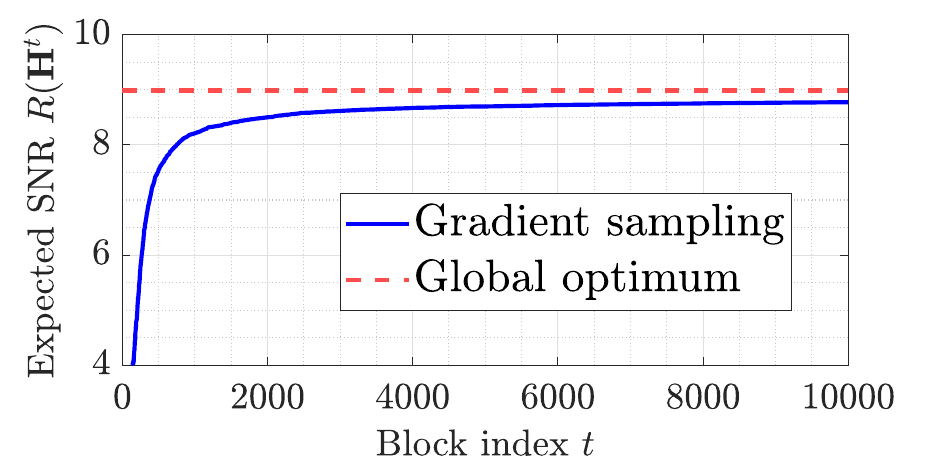}
    \caption{The toy example to show that the gradient sampling algorithm cannot converge to the global optimum of the expected SNR maximization problem \eqref{eq:actual_problem}. In this example, we set $N=3$ and $K=4$. The channels are $h_0=0$, $h_1=0.6404 + 0.3616j$, $h_2=0.1049 + 1.3040j$, and $h_3=-0.5357 + 0.9471j$. It can be observed that there remains a small gap between the converged objective value of the gradient sampling algorithm and the global optimum of problem \eqref{eq:actual_problem}.}
    \label{fig:example_converge}
\end{figure}

\begin{theorem}[Finite-time stationarity]
\label{theorem:finite}
Consider the gradient sampling update in \eqref{eq:update_H}. Let $L_R$ be the Lipschitz
constant of $R(\bH)$ 
and define
\begin{equation}
    R_{\mathrm{max}}= \max_{\bm\theta\in\Phi_K^N}\mathrm{SNR}(\bm\theta).
\end{equation}
Moreover, we define the partial derivative vector as
\begin{equation}
    \nabla R(\bH)\triangleq
\left(
\frac{\partial R(\bH)}{\partial H_n(\phi)}
\right)_{n=1,\ldots,N,\ \phi\in\Phi_K}.
\end{equation}
Then, for any finite horizon $T\geq 1$ and any positive learning rates
$\{\gamma^t\}_{t=1}^T$, we have
\begin{equation}
\min_{1\leq t\leq T}
\mathbb{E}\!\left[\|\nabla R(\bH^t)\|^2\right]
\leq
\frac{
R_{\mathrm{max}}+L_R NR_{\mathrm{max}}^2\sum_{t=2}^T(\gamma^t)^2
}{
\sum_{t=2}^T\gamma^t
}.
\end{equation}
In particular, by choosing the horizon-dependent constant learning rate, i.e., $\gamma^t=\frac{\eta}{\sqrt{T}}, t=2,\ldots,T$,
for some constant $\eta>0$, we obtain
\begin{equation}
\min_{1\leq t\leq T}
\mathbb{E}\!\left[\|\nabla R(\bH^t)\|^2\right]
\leq
\frac{R_{\max}/\eta+N R_{\max}^2\eta}{\sqrt{T}}
=
O\left(\frac{1}{\sqrt{T}}\right).
\end{equation}
\end{theorem}
\begin{IEEEproof}
The proof follows from the standard finite-time analysis of stochastic
gradient methods; see, e.g.,
\cite[Lemma 4.2 and Theorem 4.10]{bottou2018optimization}.
\end{IEEEproof}

In summary, Theorem \ref{theorem:asymptotic}  establishes the asymptotic first-order stationarity of the proposed gradient bandit algorithm under
Robbins--Monro learning rates, while Theorem \ref{theorem:finite} provides the finite-time stationarity that a horizon-dependent constant learning rate yields the $O(\frac{1}{\sqrt{T}})$ convergence rate. We remark that the convergence results in Theorems \ref{theorem:asymptotic} and \ref{theorem:finite} represent the strongest general conclusion for the proposed gradient sampling algorithm. Since the expected SNR objective induced by the product softmax probability is nonconvex, global optimality cannot be guaranteed. We illustrate this point with a toy example for the case of $N=3$ and $K=4$ in Fig. \ref{fig:example_converge}. Consequently, a regret bound analysis for the gradient sampling algorithm is not available because such an analysis would require a guarantee of global optimality.

\begin{figure}
    \centering
    \includegraphics[width=0.85\linewidth]{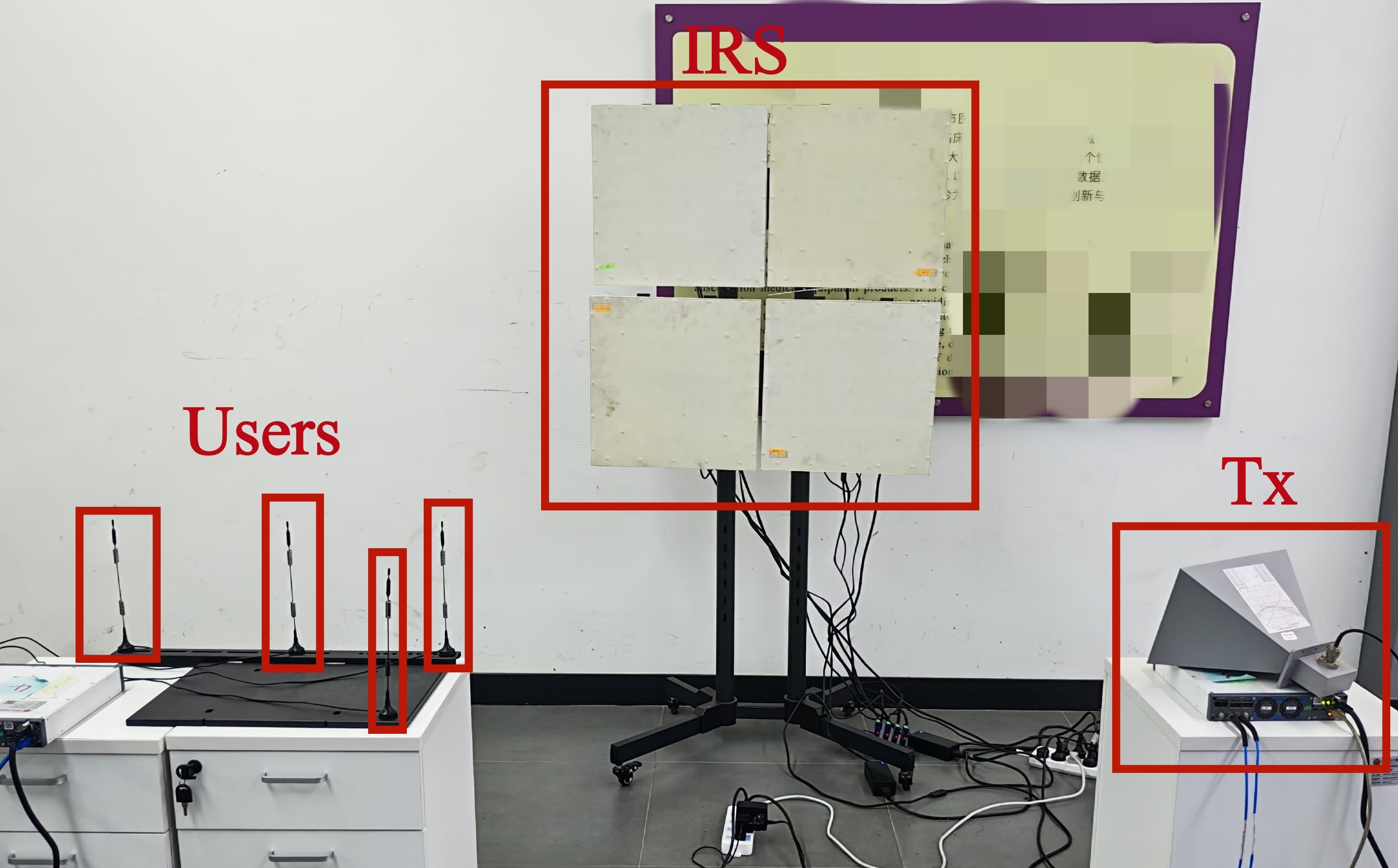}
    \caption{The real-world testbed at 3.5 GHz to evaluate the performance of different reinforcement learning algorithms.}
    \label{fig:field_test}
\end{figure}

\begin{table}[t]
\small
\renewcommand{\arraystretch}{1.5}
\centering
\caption{\small The values of $\epsilon^\star$ and $T_0^\star$ under different $T$.}
\begin{tabular}{|c|ccccc|}
\firsthline
 $T$ & $1000$ & $1500$ & $2000$ & $2500$ & $3000$ \\
\hline
 $\epsilon^\star$ & $0.498$ & $0.431$ & $0.388$ & $0.359$ & $0.336$ \\
 $T_0^\star$ & $498$ & $646$ & $776$ & $897$ & $1008$ \\
\lasthline
\end{tabular}
\label{tab:epsilog_T0_star}
\end{table}

\section{Field Tests}
\label{sec:field}

\subsection{Field Test Settings}

We conducted the field tests in an indoor scenario as shown in Fig.~\ref{fig:field_test}. Our IRS prototype machine consists of custom-designed printed circuit board (PCB) modules, each comprising a $10\times 10$ array of $100$ meta-atoms. Equipped with dual PIN diodes, each meta-atom is independently programmable to furnish four discrete phase-shift states: $\{0, \pi/2, \pi,3\pi/2\}$. The IRS operates at a center frequency of $3.5$ GHz with a $100$ MHz bandwidth. We employ the NI-USRP X410 SDR to instantiate both the transmitter and the receiver. A directional antenna is deployed at the transmitter, and a single omnidirectional antenna is deployed at the receiver. The system operates using a narrowband transmission scheme, emitting a $3.5$ GHz carrier with a $125$ KHz baseband bandwidth. The default transmit power is set to $P=5$ dBm. Baseband signals are sampled at a rate of $500$ KHz at the receiver. All the computations are conducted on a laptop with a 3.2 GHz CPU and 32 GB RAM.

\begin{figure}[t]
   \centering
    \includegraphics[width=\linewidth]{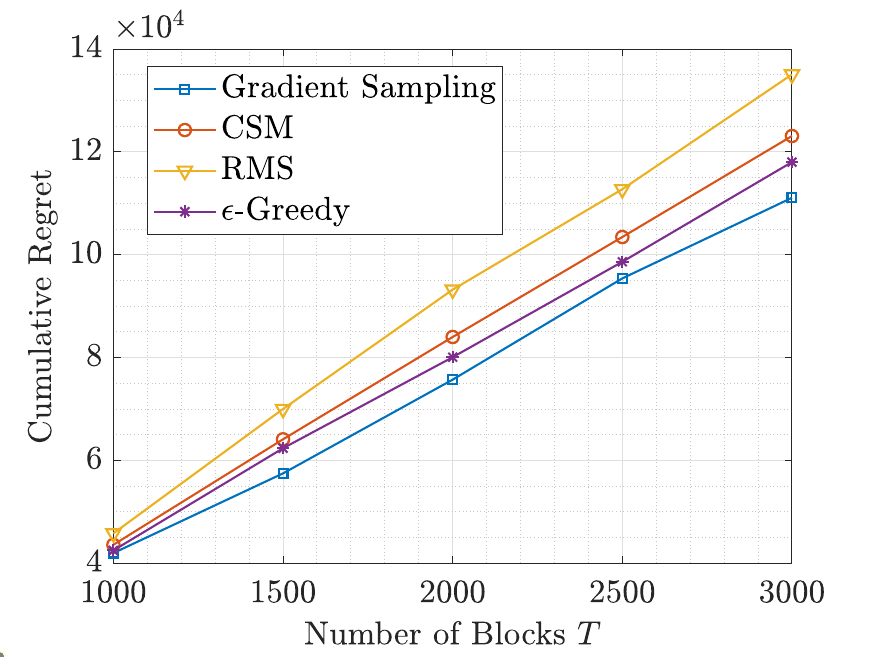}
    \caption{Cumulative regret versus the number of blocks.}
    \label{fig:CRegret}
\end{figure}

\begin{figure}[t]
   \centering
    \includegraphics[width=\linewidth]{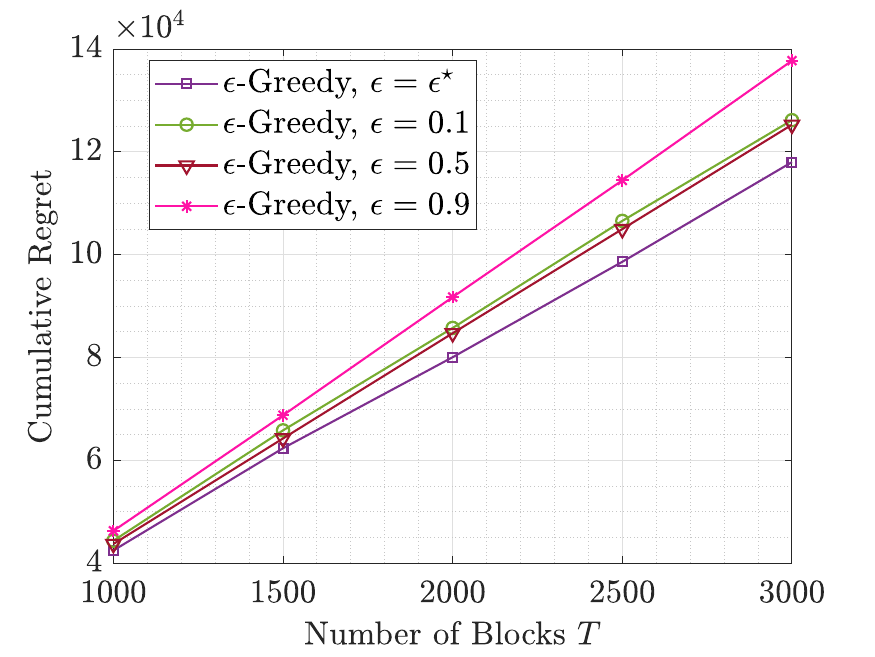}
    \caption{Comparison of different $\epsilon$ values for $\epsilon$-Greedy algorithm.}
    \label{fig:diff_epsilon}
\end{figure}

Aside from CSM, $\epsilon$-greedy, and the proposed gradient sampling algorithm, we also consider a baseline method called \emph{random max sampling (RMS)}---which replaces the exploitation step in the CSM by using the best phase shift vector in the exploration stage, i.e.,
$\bm\theta^t=\arg\max_{\bm\theta\in\{\bm\theta^1,\ldots,\bm\theta^{T_0}\}} \mathrm{SNR}(\bm\theta).$
We set the learning rate $\gamma^t=\frac{0.005}{\sqrt{T}},\;t=1,2,\ldots,T$ for the proposed gradient sampling method. Regarding the $\epsilon$-greedy algorithm, we obtain $\epsilon^\star$ for the different $T$ values according to \eqref{opt:epsilon} as shown in Table \ref{tab:epsilog_T0_star}; we shall also try out other values of $\epsilon$ for comparison purposes. For the CSM and the RMS methods, we obtain $T_0^\star$ by rounding $\epsilon^\star T$ to the nearest integer; the values of $T_0^\star$ for the different $T$ values are also shown in Table \ref{tab:epsilog_T0_star}.
Moreover, we take $Q_n^{10000}(\varphi)$ as an approximation of $Q_n^\infty(\varphi)$ to compute the cumulative regret.

\subsection{SNR Enhancement Performance}
We start by considering the SNR enhancement performance of different algorithms. Fig. \ref{fig:CRegret} shows the cumulative regrets achieved by the various methods versus the number of blocks $T$. Observe that the RMS method yields the highest cumulative regret among the reinforcement learning methods, and the performance gap gradually increases with $T$. This is because the phase shifts used by RMS in the exploitation blocks are not selected to maximize CSM. In contrast, the CSM algorithm achieves a much lower cumulative regret. In addition, $\epsilon$-greedy slightly outperforms CSM when both methods use the same optimized exploration probability $\epsilon^\star$. It is also worth noting that the gradient sampling method significantly outperforms all other methods, although its objective is not specifically designed to minimize cumulative regret.

\begin{figure}[t]
    \centering
    \includegraphics[width=\linewidth]{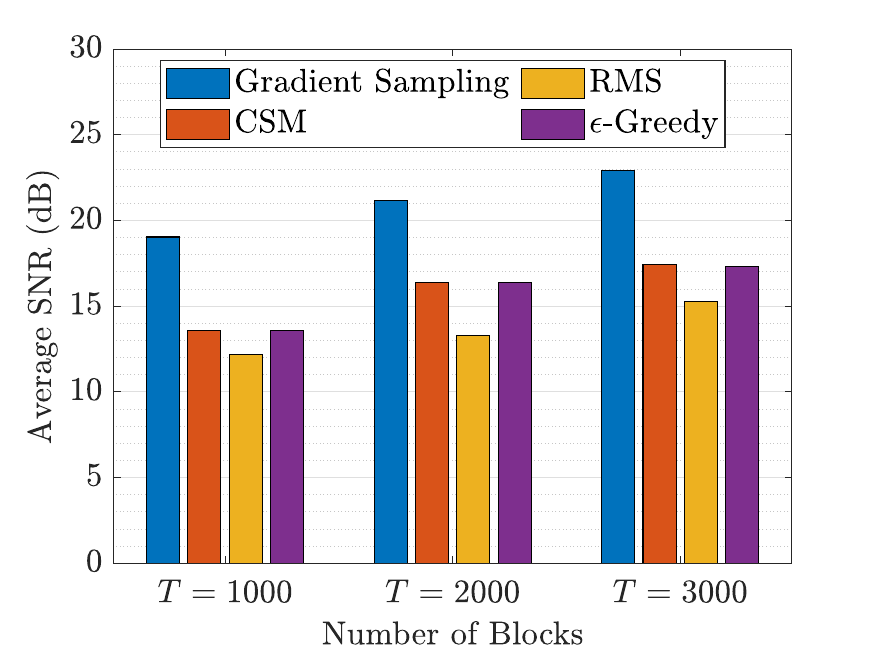}
    \caption{Average SNR versus the number of blocks.}
    \label{fig:avg_SNR}
\end{figure}

\begin{figure}[t]
    \centering
    \includegraphics[width=\linewidth]{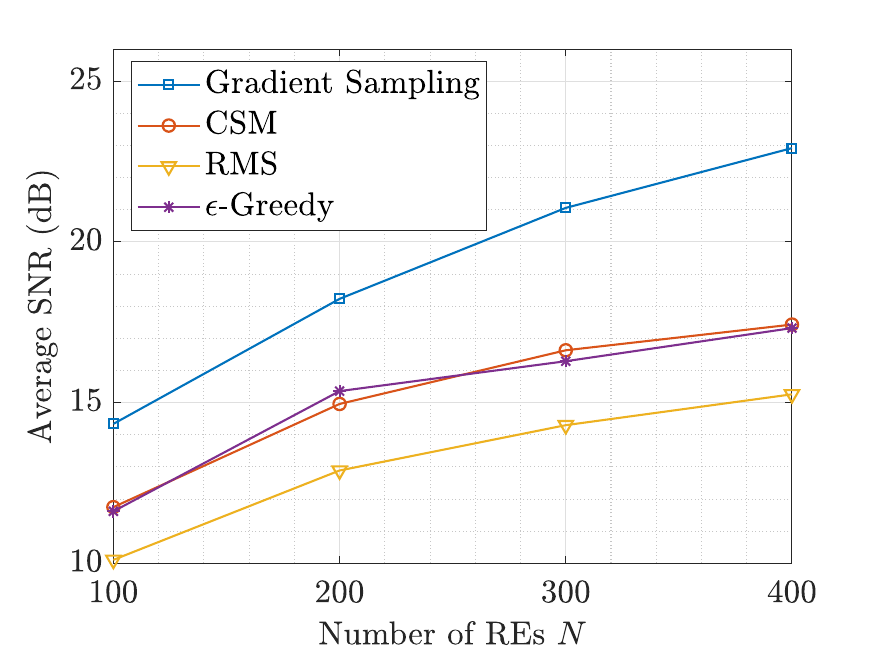}
    \caption{Average SNR versus the number of REs.}
    \label{fig:SNR_N}
\end{figure}

\begin{figure}[t]
    \centering
    \includegraphics[width=\linewidth]{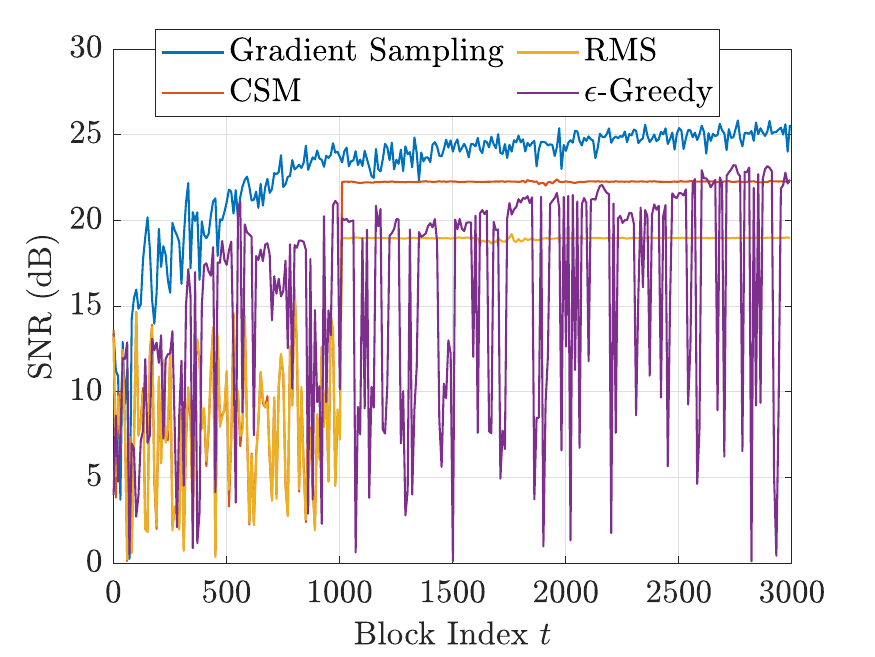}
    \caption{Visualize how the SNR changes with $t$ when $T=3000$.}
    \label{fig:SNR_visual}
\end{figure}

\begin{figure}[t]
    \centering
    \includegraphics[width=\linewidth]{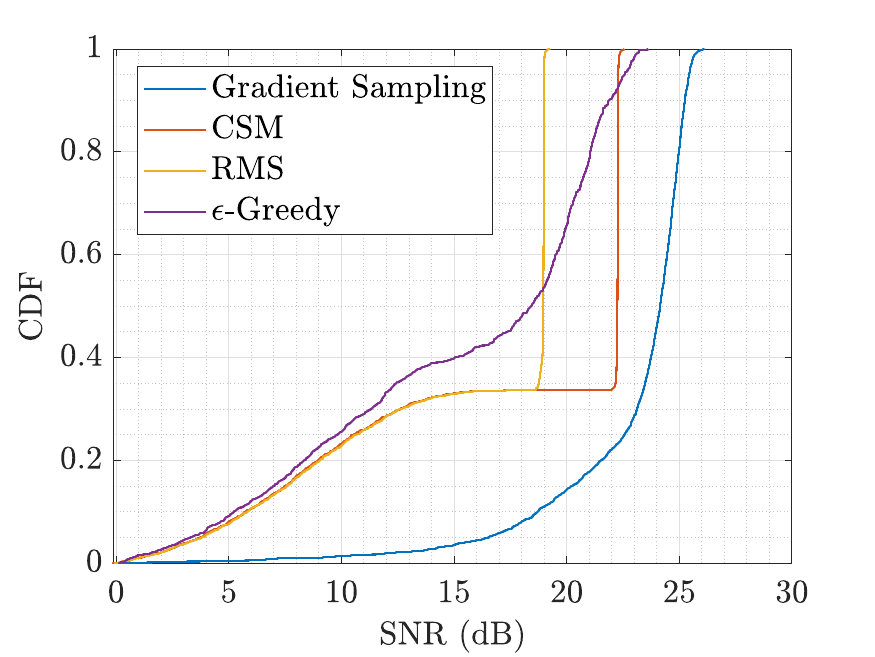}
    \caption{CDF of the SNR when $T=3000$.}
    \label{fig:cdf}
\end{figure}

Next, we examine the impact of the exploration probability on the performance of the $\epsilon$-greedy method. As shown in Fig.~\ref{fig:diff_epsilon}, the cumulative regret is the highest when $\epsilon=0.1$ and $\epsilon=0.9$, which indicates that neither excessive exploration nor excessive exploitation is desirable. The $\epsilon$-greedy method with the optimized $\epsilon^\star$ achieves much better performance than the other choices of $\epsilon$, demonstrating the effectiveness of our strategy for optimizing $\epsilon$ based on the upper-bound approximation of cumulative regret.

Fig.~\ref{fig:avg_SNR} further demonstrates the average SNR across $T$ blocks for the different methods. As $T$ increases from $1000$ to $3000$, all methods show an increasing trend in average SNR, indicating that more blocks provide more observations for identifying effective phase shifts. Among all methods, gradient sampling consistently achieves the best performance. Its average SNR increases from about $19$ dB at $T=1000$ to nearly $23$ dB at $T=3000$. In contrast, RMS achieves the lowest SNR in all cases, suggesting that the phase shifts selected by RMS are less effective for SNR enhancement. CSM and $\epsilon$-greedy achieve comparable performance, but both remain clearly below gradient sampling. For example, when $T=3000$, their average SNR values are about $5$ dB lower than gradient sampling. These results demonstrate the advantage of gradient sampling in identifying high-quality phase shifts.

Fig. \ref{fig:SNR_N} shows the average SNR achieved by different algorithms versus the number of REs $N$. It can be observed that all methods achieve a higher average SNR with more REs. Moreover, gradient sampling consistently achieves the highest average SNR across all values of $N$. Its average SNR increases from about $14$ dB at $N=100$ to about $23$ dB at $N=400$. In comparison, the performance of CSM and $\epsilon$-greedy also improves with $N$, but the improvement is less significant than that of gradient sampling. RMS achieves the lowest performance among the four methods over the entire range of $N$. These results indicate that simply increasing the number of REs is not sufficient; an effective sampling strategy is also needed to fully utilize the enlarged IRS.

Finally, Fig.~\ref{fig:SNR_visual} and Fig.~\ref{fig:cdf} illustrate the block-level SNR performance when $T=3000$. Specifically, Fig.~\ref{fig:SNR_visual} shows how the SNR changes over block index $t$, while Fig.~\ref{fig:cdf} presents the corresponding CDF of the SNR samples. As shown in Fig.~\ref{fig:SNR_visual}, gradient sampling quickly improves the SNR in the early blocks and then maintains a high SNR of about $25$ dB, indicating that it can rapidly identify high-quality phase shifts. In contrast, CSM and RMS exhibit large fluctuations during the exploration stage and then converge to relatively stable SNR values of about $22$ dB and $19$ dB, respectively. The $\epsilon$-greedy method also improves with $t$, but it suffers from frequent low-SNR blocks due to random exploration. The CDF in Fig.~\ref{fig:cdf} confirms these observations from a distributional perspective. The curve of gradient sampling is shifted farthest to the right, and most of its SNR samples are above $23$ dB, whereas CSM and RMS have a considerable portion of samples below this level. Although $\epsilon$-greedy can reach high SNR values, it has a wider SNR distribution and a longer low-SNR tail.

\begin{figure}[t]
    \centering
    \includegraphics[width=\linewidth]{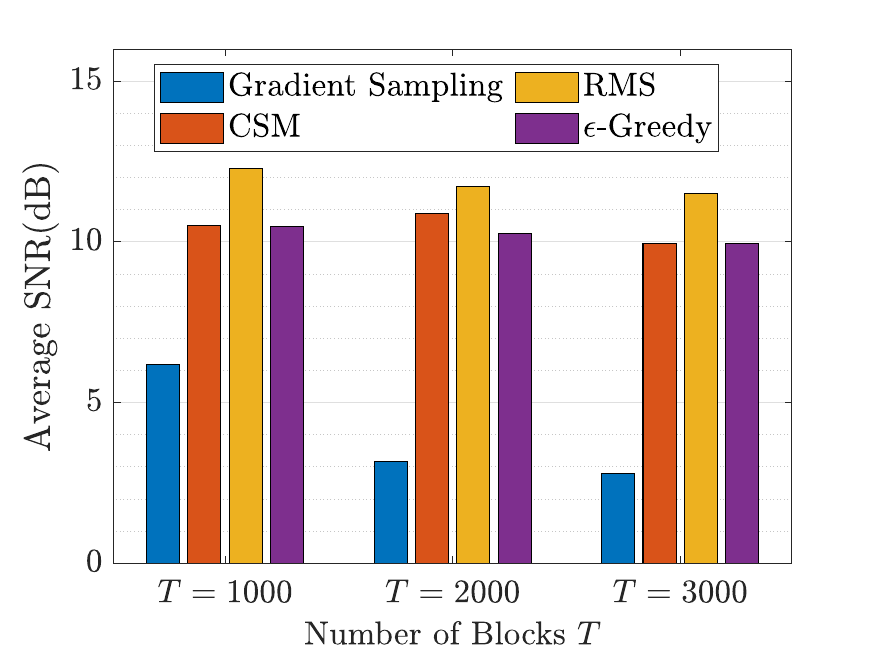}
    \caption{Average SNR versus $T$ for SNR reduction.}
    \label{fig:reduction_T}
\end{figure}

\begin{figure}[t]
    \centering
    \includegraphics[width=\linewidth]{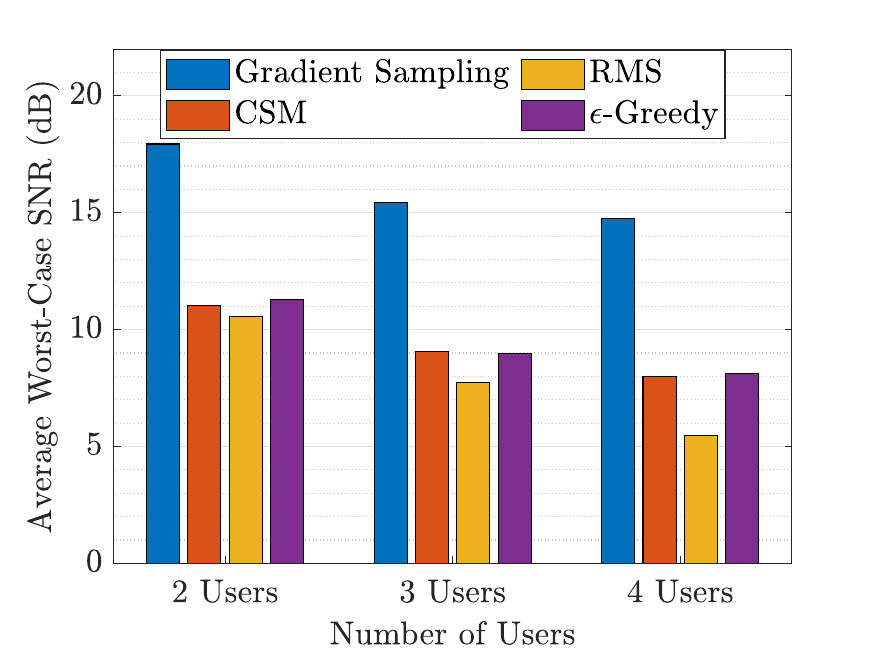}
    \caption{Average worst-case SNR versus the number of users.}
    \label{fig:SNR_U}
\end{figure}

\subsection{SNR Reduction and Max-Min Fairness Performance}

We are also interested in the performance of these reinforcement learning methods in other tasks, where SNR is replaced with the corresponding task-specific objectives. We first consider the SNR reduction task in which the reward is now given by $f(\bm\theta^t)=-\mathrm{SNR}(\bm\theta^t)$. Moreover, we increase the transmit power $P$ to $10$ dBm in this task. Unlike the SNR enhancement task, a lower average SNR indicates better performance in this setting. Fig.~\ref{fig:reduction_T} compares the average SNR achieved by different methods under the SNR reduction task. We find that gradient sampling achieves the lowest SNR across all values of $T$, and its average SNR is below $3$ dB when $T=3000$. This result shows that gradient sampling also performs well for SNR reduction. In contrast, CSM and $\epsilon$-greedy maintain much higher average SNR values, mostly around $10$--$12$ dB, and their performance changes only slightly as $T$ increases. RMS still achieves the worst performance. When $T=3000$, the average SNR of RMS is about $8$ dB higher than that of gradient sampling.

Finally, we evaluate the max-min fairness performance of these algorithms in the multi-user broadcast network with $U$ users. In this setting, the worst-case SNR among the users is the reward in each block $t$, i.e., $f(\bm\theta^t)=\min_{u}\mathrm{SNR}_{u}(\bm\theta^t)$. As shown in Fig.~\ref{fig:SNR_U}, when the number of users increases from $2$ to $4$, the average worst-case SNR of all methods decreases because it becomes more difficult to balance the need for more users simultaneously. However, gradient sampling still maintains a clear advantage over the other methods. For example, when the system serves $4$ users, gradient sampling achieves an average worst-case SNR of about $14.7$ dB, while CSM and $\epsilon$-greedy achieve about $8$ dB, and RMS drops to about $5$ dB. These results show that gradient sampling can also identify high-quality phase shifts for multi-user broadcast transmission.

\section{Conclusion}
\label{sec:conclusion}

Blind beamforming for IRS without channel estimation is of great practical importance. This work considers blind beamforming from a reinforcement learning perspective, showing that the existing CSM method \cite{Arun2020RFocus, ren2022configuring} can be recognized as a special case of the $\epsilon$-greedy algorithm for the multi-armed bandit problem. We further propose an upper approximation of the cumulative regret and thereby optimize the exploration probability $\epsilon$ in closed form. Moreover, we suggest a gradient sampling scheme that is more efficient than the uniform sampling in the $\epsilon$-greedy algorithm. The convergence of the proposed gradient sampling is also verified by showing that it is a stochastic approximation to gradient ascent. Field tests at the 3.5 GHz real-world testbed demonstrate these reinforcement learning approaches.

\bibliographystyle{IEEEtran}     
\bibliography{IEEEabrv,ref}
\end{document}